\documentclass[11pt,twoside]{article}

\usepackage{graphicx}
\usepackage{amsmath}
\usepackage{amssymb}
\usepackage{amsfonts}
\usepackage{subfig,wrapfig}

\newcommand{\term}[1]{\emph{#1}}

\newcommand{\norm}[1]{\lVert#1\rVert}

\newtheorem{theorem}{Theorem}
\newtheorem{lemma}[theorem]{Lemma}
\newtheorem{corollary}[theorem]{Corollary}

\newenvironment{proof}{\QuadSpace\par\noindent{\bf Proof}:}{\EndProof\HalfSpace}

\newenvironment{notation}{\QuadSpace\par\noindent{\bf Notation}:}{\HalfSpace}

\newcommand{\QuadSpace}{\vspace{0.25\baselineskip}}
\newcommand{\HalfSpace}{\vspace{0.5\baselineskip}}

\newcommand{\EndProof}{ \hfill \vrule width 1ex height 1ex depth 0pt }

\begin{document}

\title{A Locked Orthogonal Tree}
\author{
	David Charlton \\ Boston University \\ charlton@cs.bu.edu \and
	Erik D. Demaine \\ MIT \\ edemaine@mit.edu \and
	Martin L. Demaine \\ MIT \\ mdemaine@mit.edu \and
	Gregory Price \\ MIT \\ price@mit.edu \and
	Yaa-Lirng Tu \\ MIT \\ lingding@mit.edu
}

\maketitle

\begin{abstract}
We give a counterexample to a conjecture of Poon \cite{poon} that any orthogonal tree in two dimensions can always be flattened by a continuous motion that preserves edge lengths and avoids self-intersection. We show our example is locked by extending results on strongly locked self-touching linkages due to Connelly, Demaine and Rote \cite{cdr2} to allow zero-length edges as defined in \cite{adg}, which may be of independent interest. Our results also yield a locked tree with only eleven edges, which is the smallest known example of a locked tree.
\end{abstract}

\section{Introduction}

Connelly, Demaine and Rote \cite{cdr1} proved the ``Carpenter's Rule Theorem'': any chain or polygon in two dimensions can be convexified by a motion that preserves edge lengths and avoids self-intersection. In contrast, however, no such theorem applies to trees: Biedl et al.\ \cite{biedl} showed that there are so-called ``locked'' trees that cannot reach a canonical configuration via such a motion. (We say such trees cannot be \term{flattened}; see Figure~\ref{fig:Flatten}). This result was improved in \cite{cdr2} to show that even a single degree-3 vertex allows an otherwise chain-like graph to lock, so that (by one measure, at least) the Carpenter's Rule Theorem is optimal.

Subsequent work focused on what kinds of trees can lock. Poon \cite{poon2} showed a locked tree of diameter 4, which is minimal, since trees of diameter 3 cannot lock. In the opposite direction, Poon \cite{poon} showed that all lattice trees (that is, trees with unit-length, orthogonal edges) can be flattened, as can certain classes of diameter-4 trees. Because all locked trees discovered by that point were based on what one might call ``angular constraints,'' that is, on packing together subtrees that required local angular expansion in order to flatten, he went further, conjecturing that all orthogonal trees could be flattened.

In this paper we refute this conjecture, showing an 11-edge locked nearly orthogonal tree that can be transformed into a 21-edge locked orthogonal tree. We also extend the machinery of strongly locked self-touching linkages of \cite{cdr2} and \cite{adg}, to allow for self-touching linkages with zero-length edges, which might be useful in future analysis.

\begin{figure}
  \centering
  \includegraphics[scale=0.75]{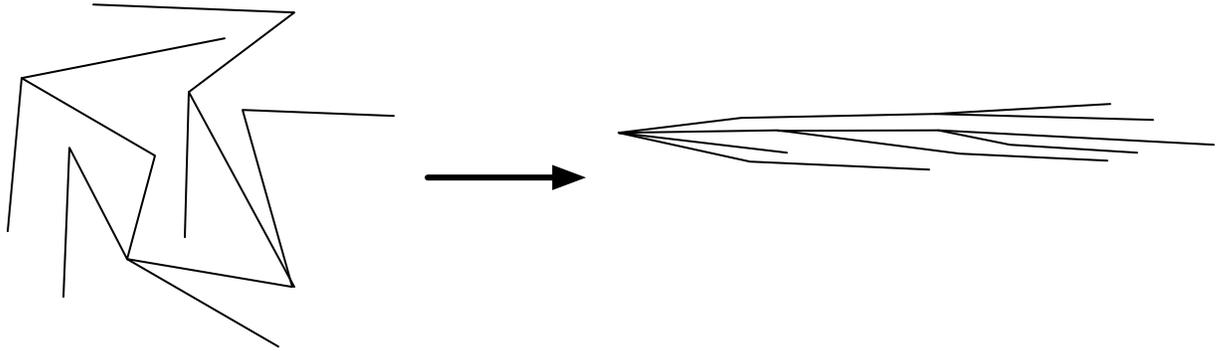}
  \caption{\label{fig:Flatten}
    Flattening a linkage: the initial tree (left) can be continuously transformed into the ``flat'' tree on the right, with all edges trailing off in the same direction from the root. Though shown distinct for clarity, all edges in the right figure lie along a single geometric line.}
\end{figure}

\section{Terminology}

A \term{(planar) linkage} is a simple graph together with an assignment of a nonnegative real length to each edge and a combinatorial planar embedding (clockwise order of edges around each vertex and which edges form the outer face). A \term{configuration} of a linkage is a (possibly self-intersecting) straight-line drawing of that graph in the plane, respecting the linkage's combinatorial embedding, such that the Euclidean distance between adjacent nodes equals the length assigned to their shared edge.

We are primarily interested in the standard case of \term{nontouching} configurations, that is, configurations in which no edges intersect each other except at a shared vertex. The set of all such configurations is called the \term{configuration space} of the linkage. A \term{motion} of a nontouching configuration $C$ is a continuous function $m$ from $[0,1]$ to the configuration space of the linkage such that $m(0) = C$.

To analyze nontouching configurations it is helpful to also consider \term{self-touching} configurations, where we relax our constraints to allow intersections, as long as edges do not actually cross each other. This substantially complicates the definitions, since edges can share the same geometric location, and we need a way to distinguish which ``side'' each edge is on in this case. We also need to be more careful in generalizing the definition of a motion, since two geometrically identical configurations may have different sets of valid motions depending on the combinatorial ordering of the edges. A full discussion of these details is beyond our scope, and we will simply rely on the formalization and results of \cite{cdr2}, \cite{cddf} and \cite{adg}. When we refer to a self-touching configuration or motion satisfying ``noncrossing constraints,'' it should usually be clear what is intended, but for the correct formalization behind that intuition see the references. The reader who wishes for a more formal intuition can think of a self-touching configuration as a convergent sequence of nontouching configurations.

A configuration $C'$ is a \term{$\delta$-perturbation} of a configuration $C$ if the position of each vertex in $C'$ differs (in Euclidean distance) from its position in $C$ by at most $\delta$. In particular, we allow the edge lengths in $C'$ to differ from those in $C$, thus altering the underlying linkage.

A configuration of a tree linkage can be \term{flattened} if it has a motion transforming it as in Figure~\ref{fig:Flatten} so that all edges are trailing off in the same direction from an arbitrarily chosen root node. Otherwise, it is \term{unflattenable}. (Which node is chosen as the root does not affect the definition; see \cite{biedl}.) A linkage is \term{locked} if its configuration space is topologically disconnected. For tree linkages, which are the only type we consider, it is equivalent to say a linkage is locked if it has an unflattenable configuration. A self-touching configuration is \term{rigid} if it has no nontrivial nonrigid motion. A configuration is \term{locked within $\varepsilon$} if no motion can change the position of any vertex by more than $\varepsilon$ (modulo equivalence by rigid motions).

A configuration $C$ is \term{strongly locked} if for any $\varepsilon > 0$ there exists a $\delta > 0$ such that any $\delta$-perturbation of $C$ that satisfies $C$'s noncrossing constraints is locked within $\varepsilon$. This definition trivially implies unflattenability, and thus also that the underlying linkage is locked. Note that to be strongly locked, $C$ must by definition be rigid and therefore self-touching. \cite{cdr2} shows that the necessary property of rigidity is also sufficient (see Theorem~\ref{thm:StronglyLocked} below).

\section{Nearly Orthogonal Tree}

\begin{figure}
  \centering
  \subfloat[][]{
    \includegraphics[scale=.85]{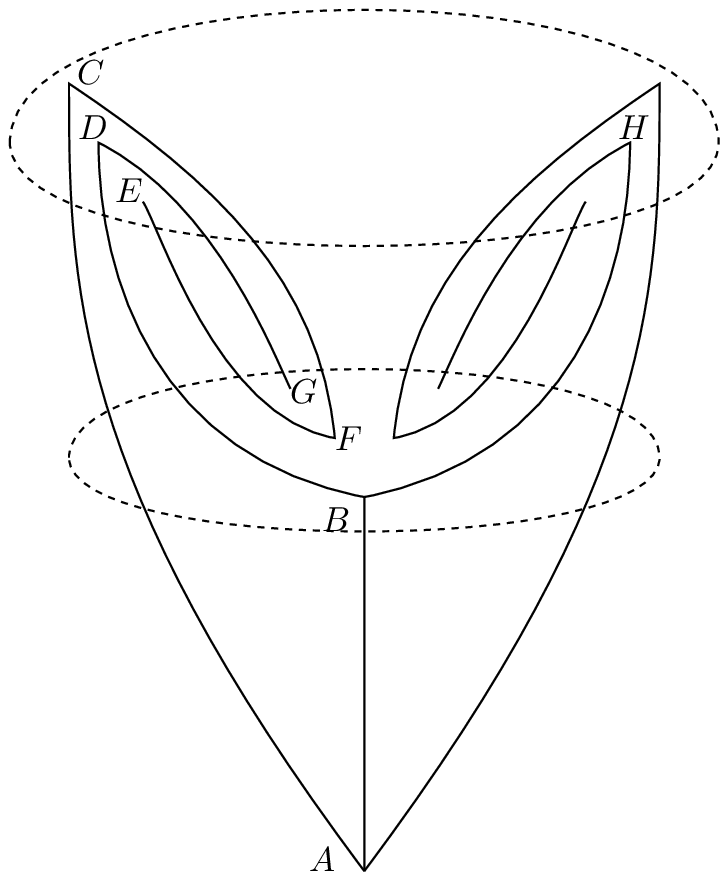}
    \label{fig:LockedTree}}\hfil
  \subfloat[][]{
    \includegraphics[scale=.85]{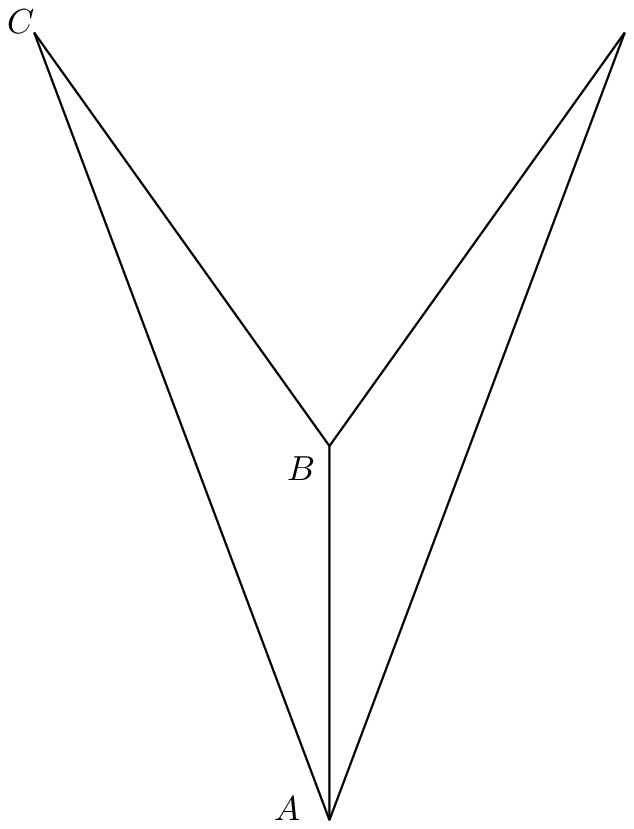}
    \label{fig:SimplifiedTree}}\hfil
  \caption{Left: A ``pulled apart'' view of our locked self-touching tree. Dotted lines surround vertices that are at the same geometric location. Right: The simplified version of the tree after applying rules from \cite{cddf}.}
\end{figure}

Consider the self-touching tree in Figure~\ref{fig:LockedTree}. The geometry of this tree is a straight vertical line with only three distinct vertices at the top, center and bottom, but it is shown ``pulled apart'' to ease exposition. We claim this tree is locked. We use two lemmas from \cite{cddf}:

\begin{lemma}[Rule 1 from \cite{cddf}] \label{lem:Rule1}
If a bar $b$ is collocated with another bar $b'$ of 
equal length, and the bars incident to $b'$ form angles less 
than $90^\circ$ on the same side as $b$, then any motion must keep 
$b$ collocated with $b'$ for some positive time.
\end{lemma}

\begin{lemma}[Rule 2 from \cite{cddf}] \label{lem:Rule2}
If a bar $b$ is collocated with an incident bar $b'$ of 
the same length whose other incident bar $b''$ forms a convex 
angle with $b'$ surrounding $b$, then any motion must keep $b$ 
collocated with $b'$ for some positive time.
\end{lemma}

Using these lemmas, we can quickly show:

\begin{theorem}
The self-touching linkage of Figure~\ref{fig:LockedTree} is rigid.
\end{theorem}

\begin{proof}
Applying Lemma~\ref{lem:Rule1} to edges $DG$ and $CF$, we see they must be collocated for positive time because of the enclosing edges $EF$ and $CA$. Furthermore, by Lemma~\ref{lem:Rule2}, $CF$ is also collocated with the adjacent edge $EF$ for positive time because of the enclosing edge $CA$. Finally, by Lemma~\ref{lem:Rule1}, $EF$ is collocated with $DB$ for positive time because of enclosing edges $DG$ and $BH$. Thus, all four edges mentioned are collocated for positive time, and since the figure is symmetric an identical analysis applies on the opposite side of the figure.

The preceding implies that for positive time, our figure reduces to the trivially rigid linkage of Figure~\ref{fig:SimplifiedTree}, so we are done.
\end{proof}

One more earlier result will give us what we need:

\begin{theorem}[Theorem 8.1 from \cite{cdr2}] \label{thm:StronglyLocked}
Any rigid self-touching configuration is strongly locked.
\end{theorem}

\begin{corollary}
The tree in Figure~\ref{fig:LockedTree} is strongly locked.
\end{corollary}

Thus, by a small perturbation to obtain a nontouching tree, the figure yields a new example of (an unflattenable configuration of) a locked tree. However, such a perturbation will not yield an orthogonal tree. We address this issue in the next section.

\section{Fully Orthogonal Tree}

\begin{figure}
  \centering
  \includegraphics[scale=0.75]{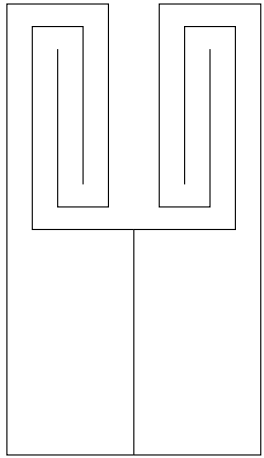}
  \caption{\label{fig:OrthoTree}
    The orthogonal version of the tree. This configuration is locked when the dimensions are chosen appropriately.}
\end{figure}

A simple modification of our tree makes it orthogonal: see Figure~\ref{fig:OrthoTree}. We will show that this diagram is in fact unflattenable (if the dimensions are chosen appropriately), but the proof is not immediate. The general idea is that this configuration can still be viewed as a ``small perturbation'' of the original tree from Figure~\ref{fig:LockedTree}, if we add zero-length edges to the original figure wherever Figure~\ref{fig:OrthoTree} has a horizontal edge. Then, as long as the horizontal edges are sufficiently short, they can still be viewed as a ``small perturbation'' of a zero-length edge.

Unfortunately, Theorem~\ref{thm:StronglyLocked} no longer works directly when the self-touching configuration has zero-length edges. To cover this situation we need the more recent work of Abbott, Demaine and Gassend \cite{adg} that extends the machinery of self-touching configurations in order to prove a generalized Carpenter's Rule Theorem. In this more general model, zero-length edges are permitted.

\begin{theorem}
Theorem~\ref{thm:StronglyLocked} still holds when zero-length edges are allowed.
\end{theorem}

In the proof we will explicitly use the notation and results of \cite{cdr2} and \cite{adg}. In lieu of providing complete coverage of the requisite background material, we also give a brief sketch of the arguments involved.

\begin{proof} (Sketch)
The original proof from \cite{cdr2} uses a topological argument. The idea is that since a rigid configuration is an isolated point in its (self-touching) configuration space, and using the fact that this configuration space is locally defined by semialgebraic constraints in its vertex coordinates, a small perturbation in those coordinates can yield only a small expansion to the connected component.

The generalized definitions we need from \cite{adg} still define the configuration space by semialgebraic constraints, so a similar argument still holds, and we need only show that we can substitute the new definitions in the topological arguments of \cite{cdr2}.
%
\end{proof}

\begin{proof} (Full)
\begin{notation}
We reuse the terminology and notation of \cite{cdr2}, which we briefly review here: for a linkage $L$, we use $\text{NConf}_\delta(L)$ to denote the space of all non-self-touching configurations of all linkages $\delta$-related to $L$ (that is, roughly, all linkages whose graphs closely approximate that of $L$). An \term{annotated} configuration of $L$ is a configuration of $L$ together with, for each pair of edges, an extra (real) value indicating which ``side'' of each other they are on. These values are discontinuous at points where edges overlap for positive length, but are continuous for all non-self-touching configurations, and thus can be used to define the combinatorial structure of a self-touching configuration as a limit of non-self-touching annotated configurations. The space of annotated non-self-touching configurations of all linkages $\delta$-related to $L$ is denoted $\text{Annot}_L(\text{NConf}_\delta(L))$. The set of all limit points of sequences within $\text{Annot}_L(\text{NConf}_1(L))$ is denoted $A(L)$. $A(L)$ is thus the real space within which self-touching linkages can move, incorporating both geometric data and combinatorial ordering of edges.
\end{notation}

We use the following topological lemma, which is the central result used in the original proof of Theorem~\ref{thm:StronglyLocked}:

\begin{lemma}[Lemma~8.2 from \cite{cdr2}] \label{lem:LimitPoint}
Let $A_\delta \subseteq \mathbb{R}^m$ ($\delta \geq 0$) be a family of closed sets with $A_\delta \subseteq A_{\delta '}$ for $0 \leq \delta < \delta '$ and
\begin{equation*}
\bigcap_{\delta > 0} A_\delta = A_0 
\end{equation*}
For $p \in A_\delta$ we denote by $B_\delta(p)$ the set of points which are reachable by a curve in $A_\delta$ starting at $p$. Let $p^* \in A_0$, suppose that the set $B_0 := B_0(p)$ is compact, and there is a positive lower bound on the distance between $B_0$ and any point in $A_0 - B_0$.

Then for every $\varepsilon > 0$ there is a $\delta > 0$ with the following property: $\norm{p - p^*} < \delta$ implies that $B_\delta(p)$ is contained in an $\varepsilon$ neighborhood of $B_0$.
\end{lemma}

What we need is to show that this lemma still applies in the context of the configuration space as defined by \cite{adg}.

Let $L$ be a linkage, and $C$ a rigid self-touching configuration of $L$. Fix the source and direction of an arbitrary (positive-length) edge of $L$, to disallow rigid motions. We have then that $C$ is an isolated point in $A(L)$. 

Now, for the lemma, set $A_0 = A(L)$ and $A_\delta = \overline{\text{Annot}_L(\text{NConf}_\delta(L))}$. It is trivial that $A_\delta \subseteq A_{\delta '}$ for $\delta < \delta '$. We also have by Lemma~3 of \cite{adg} that $A_0 = \cap_{\delta > 0} A_\delta$, as required. Set $p^* = C$. Since $C$ is an isolated point, $B_0 = \{C\}$ is compact and has positive distance to $A_0 - B_0$. The preconditions for Lemma~\ref{lem:LimitPoint} are satisfied, and Theorem~\ref{thm:StronglyLocked}, which relies on the lemma, now also follows for the more general configuration space of \cite{adg}.
%
%
%
%
\end{proof}

\begin{corollary}
The tree in Figure~\ref{fig:LockedTree} is still strongly locked when we add zero-length edges at each vertex. In particular, the orthogonal tree in Figure~\ref{fig:OrthoTree} is locked when the vertical gaps and horizontal edge lengths between vertices at the top, center and bottom are chosen to be sufficiently small.
\end{corollary}

\section{Conclusion}

We generalized the results of \cite{cdr2} to include zero-length edges. We used this result to prove an example of an unflattenable orthogonal tree, refuting a conjecture of Poon \cite{poon}. However, this leaves open another conjecture from \cite{poon} that any non-self-touching unit-length tree can be flattened.

Figure~\ref{fig:LockedTree} is the smallest known example of a locked tree, using only 11 edges. It is an interesting open question whether this is minimal, that is, whether any tree using fewer than 11 edges can lock.



\bibliography{locked_linkage}

\newcommand{\etalchar}[1]{$^{#1}$}
\begin{thebibliography}{CDD{\etalchar{+}}06}

\bibitem[ADG07]{adg}
Timothy~G. Abbott, Erik~D. Demaine, and Blaise Gassend.
\newblock A generalized carpenter's rule theorem for self-touching linkages.
\newblock Preprint, December 2007.

\bibitem[BDD{\etalchar{+}}02]{biedl}
Therese Biedl, Erik~D. Demaine, Martin~L. Demaine, Sylvain Lazard, Anna Lubiw,
  Joseph OÕRourke, Steve Robbins, Ileana Streinu, Godfried Toussaint, and Sue
  Whitesides.
\newblock A note on reconfiguring tree linkages: Trees can lock.
\newblock {\em Discrete Applied Mathematics}, 117(1--3):293--297, 2002.

\bibitem[CDD{\etalchar{+}}06]{cddf}
Robert Connelly, Erik~D. Demaine, Martin~L. Demaine, S\'andor Fekete, Stefan
  Langerman, Joseph S.~B. Mitchell, Ares Rib\'o, and G\"unter Rote.
\newblock Locked and unlocked chains of planar shapes.
\newblock In {\em Proceedings of the 22nd Annual ACM Symposium on Computational
  Geometry}, pages 61--70, Sedona, Arizona, June 2006.

\bibitem[CDR02]{cdr2}
Robert Connelly, Erik~D. Demaine, and G\"unter Rote.
\newblock Infinitesimally locked self-touching linkages with applications to
  locked trees.
\newblock In J.~Calvo, K.~Millett, and E.~Rawdon, editors, {\em Physical Knots:
  Knotting, Linking, and Folding of Geometric Objects in $R^3$}, pages
  287--311. American Mathematical Society, 2002.

\bibitem[CDR03]{cdr1}
Robert Connelly, Erik~D. Demaine, and G\"unter Rote.
\newblock Straightening polygonal arcs and convexifying polygonal cycles.
\newblock {\em Discrete \& Computational Geometry}, 30(2):205--239, September
  2003.

\bibitem[Poo05]{poon2}
Sheung-Hung Poon.
\newblock On straightening low-diameter unit trees.
\newblock In Patrick Healy and Nikola~S. Nikolov, editors, {\em Graph Drawing},
  volume 3843 of {\em Lecture Notes in Computer Science}, pages 519--521.
  Springer, 2005.

\bibitem[Poo06]{poon}
Sheung-Hung Poon.
\newblock On unfolding lattice polygons/trees and diameter-4 trees.
\newblock In {\em Proceedings of the 12th Annual International Computing and
  Combinatorics Conference (COCOON)}, pages 186--195, 2006.

\end{thebibliography}
\bibliographystyle{alpha}


%
%
%
%
%
%
%

\end{document}